\documentclass[journal,twoside,web]{ieeecolor}
\usepackage{generic}
\usepackage{cite}
\usepackage{amsmath,amssymb,amsfonts}
\usepackage{algorithmic}
\usepackage{graphicx}
\usepackage{textcomp}
\usepackage{hyperref}

\usepackage{cite}

\usepackage{accents}
\usepackage{xcolor}
\usepackage{comment}
\usepackage{float}
\usepackage[ruled,vlined]{algorithm2e}
\usepackage{graphicx}
\usepackage{epstopdf}

\newtheorem{proposition}{\textbf{Proposition}}
\newtheorem{problem}{\textbf{Problem}}

\newtheorem{theorem}{\textbf{Theorem}}
\newtheorem{definition}{\textbf{Definition}}
\newtheorem{lemma}{\textbf{Lemma}}
\newtheorem{assumption}{\textbf{Assumption}}

\def\BibTeX{{\rm B\kern-.05em{\sc i\kern-.025em b}\kern-.08em
    T\kern-.1667em\lower.7ex\hbox{E}\kern-.125emX}}
\begin{document}
\title{Finite sample guarantees for quantile estimation: An application to detector threshold tuning}
\author{David Umsonst, Justin Ruths, and Henrik Sandberg
\thanks{Submitted: 09 July 2021.
This work was supported in part by the Swedish Research Council (grant 2016-00861), and the Swedish Civil Contingencies Agency (grant MSB 2020-09672).}
\thanks{David Umsonst and Henrik Sandberg are with the Division of Decision and Control Systems in the School of Electrical Engineering and Computer Science at the KTH Royal Institute of Technology, 10044 Stockholm, Sweden (e-mail: \{umsonst,hsan\}@kth.se). }
\thanks{Justin Ruths is with the Department of Mechanical Engineering, The University of Texas at Dallas, 800 W. Campbell Rd, Richardson, TX, USA. (e-mail: jruths@utdallas.edu).}}

\maketitle

\begin{abstract}
In threshold-based anomaly detection, we want to tune the threshold of a detector to achieve an acceptable false alarm rate.
However, tuning the threshold is often a non-trivial task due to unknown detector output distributions.
A detector threshold that provides an acceptable false alarm rate is equivalent to a specific quantile of the detector output distribution.
Therefore, we use quantile estimators based on order statistics to estimate the detector threshold.
The estimation of quantiles from sample data has a more than a century long tradition and we provide three different distribution-free finite sample guarantees for a class of quantile estimators. 
The first is based on the Dworetzky-Kiefer-Wolfowitz inequality, the second utilizes the Vysochanskij-Petunin inequality, and the third is based on exact confidence intervals for a beta distribution.
These guarantees are then compared and used in the detector threshold tuning problem.
We use both simulated data as well as data obtained from an experimental setup with the Temperature Control Lab to validate the guarantees provided.
\end{abstract}


\section{Introduction}
\label{sec:introduction}
In a highly automated society the automatic detection of anomalies is of utmost importance.
The failure of detecting anomalies can have dire consequences, especially when the anomaly occurs in infrastructures critical to our everyday life such as power grids and water distribution networks.
Two notable incidences of undetected anomalies in critical infrastructures are the Northeast Blackout in 2003 \cite{NortheastBlackout2003}, where the software did not notify the operators about an anomaly, which then led to a cascading failure of the power grid, and the attack on the Ukrainian power grid \cite{UkraineAttack}, where attackers managed to take over a distribution power grid.

A detector needs to not only be able to detect anomalies, but to not trigger on nominal behavior.
Alarms during nominal behavior are called \emph{false alarms}. False alarms increase the cost of detectors and make detectors unreliable.
For example, in the survey \cite{AnaesthesistSurvey}, out of 460 anaesthesists who stated that they deliberately turned off an alarm device, $68.2\,\%$ named too many false alarms as a reason for turning off the alarm device.

Therefore, both the detection rate of anomalies as well as the false alarm rate need to be taken into account when tuning the anomaly detector.
One tool to evaluate the performance of a detector is the Receiver-Operator-Characteristic (ROC) curve \cite{TutorialForROCCurve}.
The ROC curve plots the detection rate over the false alarm rate for different detector tuning and the higher the detection rate is for a smaller false alarm rate the better the detector performs.
With the emerging threat of cyber-attacks on cyber-physical systems in recent years, Urbina \emph{et al.} \cite{Alvaro} argue that the impact of an attacker should also be taken into account when tuning the detector threshold.

Often we do not have exact knowledge about the statistics of the nominal and the anomalous behavior, but we have access to data which can be used to both design and tune the detector. Three different approaches are presented in \cite{AnomalyDetectionSurvey} for this task.
The first approach is \textit{supervised learning}, where data from both nominal and anomalous behavior are used to tune the detector. 
The second approach is \textit{unsupervised learning}, where data is available but the algorithm has to determine, what is anomalous and what is normal behavior.
The third approach provides a middle ground, since it is a \textit{semi-supervised learning} approach, which uses only nominal data to tune the detector.
A method to evaluate the performance of semi-supervised tuning of detectors is proposed in \cite{Goix2016HowTE}.

In this work, we utilize a semi-supervised learning approach and use independent and identically distributed (i.i.d.) samples of the detector output under nominal behavior to estimate a detector threshold that leads to a false alarm rate, which is close to the acceptable false alarm rate with a high probability.
Since the threshold that guarantees a pre-defined acceptable false alarm rate is equivalent to a certain quantile of the detector output distribution, we utilize a sample-based quantile estimator to estimate a detector threshold.
These quantile estimators typically use one or two order statistics of the sample data to determine the quantile and are simple to implement.
Therefore, the quantile estimators are commonly used in statistical software packages \cite{QuantileEstimationWithData}.

It is not clear though how many samples of the detector output are needed to be close to the acceptable false alarm rate, when using a threshold estimate based on sample data.
Therefore, the contribution of this work is two-fold.
First, we provide three different finite guarantees to determine the sample size needed to be close to the acceptable false alarm rate with high probability. 
The first finite guarantee, which we proposed in \cite{UmsonstACC21}, is based on the Dvoretzky-Kiefer-Wolfowitz (DKW) inequality \cite{DKWInequality}, the second finite guarantee is based on the Vysochanskij-Petunin inequality \cite{VysPetInequality}, and the third finite guarantee is based on exact confidence intervals of beta random variables \cite{ExactConfidenceIntBinomial}.
\emph{All three finite guarantees are based on samples from the detector output only and the results are distribution-free and independent of how the anomaly detector determines its output.}
Since we use a quantile estimator to estimate the threshold these distribution-free finite guarantees are also finite guarantees for the estimation of quantiles.
Second, we perform a thorough validation of the finite guarantees with both simulated, and real data obtained from an experimental setup.

In the literature, it is quite common to use nominal data or make assumptions on the nominal behavior when determining the threshold that guarantees an acceptable false alarm rate. For example, under the assumption of a Gaussian distribution for the detector input, Murguia \emph{et al.} \cite{RuthsMultivariate} give a closed-form solution to tune a $\chi^2$ detector and approximations of solutions to a CUSUM detector with resetting for an acceptable false alarm rate.
Since the true nominal detector distribution is often not known, in \cite{DistributionallyRobustTuningLCSS21} a distributionally robust approach is proposed, which makes assumptions on the finiteness of the moments of the input to a $\chi^2$ detector.

More recently attention has turned to sample-based methods that can detect anomalies without requiring a formal model of behavior. Although having more samples is intuitively better, it is important to establish the minimum number of samples necessary for detector tuning so that the detector threshold might be adjusted adaptively over time. 
Sample guarantees also provide characterizations that show how detection confidence can be improved if the detector has access to more than the minimum number of samples required.
Li \emph{et al.} \cite{WassersteinDetector} propose a new detector based on the Wasserstein distance, which uses a sample-based tuning approach to achieve an acceptable false alarm rate. 
The tuning method uses the detector inputs under nominal behavior and assumes a light-tailed distribution for the detector inputs.
Another approach to learn detectors from nominal behavior based on M-estimation is provided in \cite{LearningDetectors}.
Our approach has the advantage that no certain detector structure needs to be assumed and no knowledge about distributions is necessary such as in \cite{RuthsMultivariate,DistributionallyRobustTuningLCSS21,WassersteinDetector}, because the approach is purely based on samples.

\textit{Notation:} Let $\mathbb{R}$ and $\mathbb{Q}$ denote the set of real and rational numbers, respectively. 
We call $\gamma=\frac{n_1}{n_2}$ the irreducible fraction of $\gamma\in \mathbb{Q}$ if and only if $n_1$ and $n_2$ are coprime integers.
Let $x\in\mathbb{R}$, then $|x|$, $\lceil x\rceil$, and $\lfloor x\rfloor$ denote the absolute value of $x$, the smallest integer larger than or equal to $x$, and the largest integer smaller than or equal to $x$, respectively.
Given a set $\lbrace x_i\rbrace_{i=1}^N$, the $i$th order statistic, $x_{(i)}$ is the $i$th largest element in $\lbrace x_i\rbrace_{i=1}^N$, such that the set of order statistics $\lbrace x_{(i)}\rbrace_{i=1}^N$ is $\min_i x_{i}=x_{(1)}\leq x_{(2)}\leq\ldots\leq x_{(N)}=\max_i x_{i}$.
A random variable $X$ that follows a beta distribution with parameters $m$ and $n$ is denoted as $X\sim \mathrm{Beta}(m,n)$.
Given an event $E$, its probability, expected value, variance, and the indicator function of $E$ are given by $\mathrm{prob}\lbrace E\rbrace$, $\mathbb{E}\lbrace E \rbrace$, $\mathrm{Var}\lbrace E\rbrace$, and  $\mathbf{1}_E$, respectively.
\section{Problem formulation}
\label{sec:ProbForm}
In this section, we present the detector tuning problem, how it relates to the quantile of a random variable, and formulate the problem of determining a sample-based threshold, which with high probability guarantees only a small deviation from the acceptable false alarm.
\subsection{Tuning detector thresholds}
The problem of anomaly detection occurs in many different fields. 
For example, in healthcare when devices monitor a patient or in governmental agencies to detect tax fraud.
In this work, we look at anomaly detection in the context of a control system, where the feedback system is equipped with an anomaly detector on the controller side (see Fig.~\ref{fig:BlockDiagAbstractWithDetectorOutput}).
The input to the anomaly detector can depend both on the measurements received as well as the actuator signals determined by the controller, which are not necessarily scalar variables.
For example, if a Kalman filter is used the input to the anomaly detector are the actuator signals and the sensor measurements, which are used to determine the difference, $r$, between the received and predicted measurements. 
This difference can be further processed to determine the output $y_D$, e.g., $y_D=r^\top r$.
\begin{figure}
\centering
\includegraphics[scale=0.4]{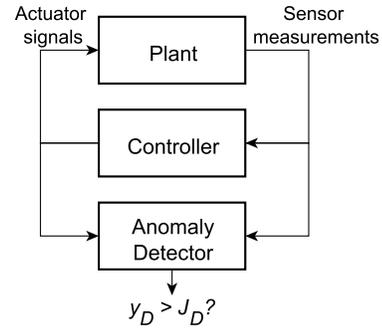}
\caption{A plant is controlled with a controller and the controller side is equipped with an anomaly detector that utilizes both the measurements and the actuator signals to determine its output $y_D$.}
\label{fig:BlockDiagAbstractWithDetectorOutput}
\end{figure}
In the control system example, a small detector output $y_D\in\mathbb{R}$ typically indicates that the system works as predicted, while large outputs indicate an unpredicted behavior.
However, in other applications, the detector output can also be a similarity measure, where a small value for $y_D$ indicates anomalous behavior (see, for example, \cite{LearningDetectors}).
In this work, we use the interpretation that a large output indicates anomalous behavior.
Therefore, an alarm is triggered when $y_D$ exceeds a threshold $J_D\in\mathbb{R}$, i.e., $y_D>J_D$, and no alarm is triggered when $y_D\leq J_D$.
Due to random processes, such as measurement noise, the detector output is also of a stochastic nature.
\begin{assumption}
\label{assum:PDFofDetectorOutput}
In the nominal case, the detector output, $y_D$, is a random variable with a continuous cumulative distribution function $F(y_D)$.
\end{assumption}
For the sake of simplicity, we assume a continuous CDF for the detector output.
In an industrial process control context, the process is often controlled around a desired steady-state value, which can be seen as stationary. Hence, assuming that the detector output is a random variable with a constant distribution is a reasonable choice.
Furthermore, if the plant has linear dynamics and the noise processes are Gaussian, a Kalman filter converges to a stationary distribution for its residual signals, which are often used as the input of the anomaly detector. If the anomaly detector has no internal dynamics, such as a neural network or norm-based detector, then the output of the detector is also a random variable with a fixed distribution.

The stochasticity of $y_D$ leads to alarms in the nominal case, so called \textit{false alarms}, where the rate of false alarms depends on the threshold. 
Since false alarms incur unnecessary costs and too many false alarms make a detector unreliable, we can choose a large threshold to avoid too many false alarms. 
However, a too large threshold leads usually to a smaller detection rate of anomalies.
Hence, there is a trade-off between the detection rate of anomalies and the false alarm rate in the nominal case when tuning the threshold.
Furthermore, Urbina \emph{et al.} \cite{Alvaro} point out that there is also a trade-off between the impact an attacker that wants to keep $y_D$ below the threshold can have and the false alarm rate when choosing the detector threshold.

In addition to that, the nature of anomalies is in most cases unknown. 
Take, for example, a complex large-scale system such as the power grid, where many different types of anomalies, such as sensor and generator failures or attacks, can occur.
Therefore, we often do not know which anomalies will occur and what detection rate we will obtain for a certain threshold $J_D$.
Since it is simpler to focus only on the nominal case instead of trying to consider all possible anomalous behaviors, we will focus on finding a threshold that guarantees an acceptable false alarm rate.
More specifically, we want to find the \emph{smallest} threshold $J_D$ such that
\begin{align}
	\label{eq:ThresholdProblem}
	\mathrm{prob}\lbrace y_D\leq J_D\rbrace\geq \gamma
\end{align}
holds, where $\gamma\in(0,1)$.
The threshold $J_D$ will result in a false alarm probability of at most $1-\gamma$.
We want to find the smallest threshold, because a trivial solution to guaranteeing an acceptable false alarm rate is to choose an arbitrarily large threshold, which in turn will also reduce the detection rate of anomalies.

In special cases, a closed form solution for the threshold exists, see, for example, \cite{RuthsMultivariate}, or the detector threshold can be approximated in a distributionally robust fashion, see \cite{DistributionallyRobustTuningLCSS21}.
If the output of the detector depends on the threshold  as well, such as for the CUSUM and MEWMA detector with resetting, it is more difficult to determine a threshold that guarantees a desired false alarm rate.
Detectors whose output depends on the threshold are one avenue of future work and will not be considered in this work.
\begin{assumption}
	The detector output $y_D$ does not depend on the threshold $J_D$.
\end{assumption}

\subsection{The problem of guaranteeing a false alarm rate}
Before we formulate the problem we consider, we want to define the notion of a $\gamma$-quantile.
\begin{definition}
	\label{def:GammaQuantile}
	The $\gamma$-quantile $J_D$ of a probability distribution is defined as
	\begin{align*}
		J_D=\inf\lbrace y_D: F(y_D)\geq \gamma\rbrace,
	\end{align*}
	where $\gamma\in(0,1)$.
\end{definition}

Note that $J_D$ in Definition~\ref{def:GammaQuantile}  is unique and finite because of the infimum operator and the fact that $\gamma\in(0,1)$.
Since $F(J_D)=\mathrm{prob}\lbrace y_D\leq J_D\rbrace$, we can immediately see that the threshold we are looking for in \eqref{eq:ThresholdProblem} is equivalent to the $\gamma$-quantile of the detector output distribution.
Further, $\gamma$-quantiles can be obtained as optimizers of convex optimization problems (see \cite{CVaRRockafellarGeneral}), which theoretically makes them efficient to calculate.
However, it is often not possible to find an expression for the probability distribution of the detector output. A reason for that is that either the plant dynamics, the controller dynamics, the detector dynamics or all of them are highly nonlinear, such that the distribution of $y_D$ does not have a closed-form solution.
Samples from the monitored process under nominal conditions are usually available, such that we can obtain samples from the detector output under nominal conditions. 
Due to the equivalence with $\gamma$-quantiles,  we use sample-based quantile estimators for a sample-based estimate $\tilde{J}_D$ of the threshold $J_D$.
In this work, we use $N$ independent and identically distributed (i.i.d.) samples of $y_D$, $\lbrace y_{D,i}\rbrace_{i=1}^{N}$, to estimate the detector threshold as
\begin{align}
	\label{eq:ApproximatedThreshold}
	\tilde{J}_D=\beta y_{D,(\lfloor N\gamma\rfloor)}+(1-\beta) y_{D,(\lfloor N\gamma\rfloor+1)},
\end{align}
where $\beta\in(0,1)$.
Note that \eqref{eq:ApproximatedThreshold} has the form of common quantile estimators used in software packages \cite{QuantileEstimationWithData}.
In our previous work, we showed how the quantile estimators can be derived from a sample approximation of the conditional Value-at-Risk \cite{UmsonstACC21}.

Although the true threshold $J_D$ can be approximated via \eqref{eq:ApproximatedThreshold}, an open problem is to determine how many samples we need to get a good approximation.
A good approximation is often characterized by assuming that $|J_D-\tilde{J}_D|$ is small with high probability. Distribution-free bounds on the bias of quantile estimates can be found in \cite{DistributionFreeBiasBounds}, where these bounds are always fulfilled and not only with high probability.
However, depending on the shape of the probability distribution even small changes from the the threshold $J_D$ can lead to large deviations in the false alarm rate. Therefore, we are more interested in how close the false alarm rate $1-F(\tilde{J}_D)$ is to the acceptable false alarm rate $1-F(J_D)=1-\gamma$, since the acceptable false alarm rate is an important variable for the system operator.
\begin{problem}
\label{prob:FiniteGuarantees}
Determine the number, $N$, of i.i.d. detector output samples needed such that
\begin{align*}
	\mathrm{prob}\lbrace |F(\tilde{J}_D)-\gamma|\leq \epsilon\rbrace \geq 1-\rho
\end{align*}
holds, where $\gamma\in(0,1)$, $\epsilon\in(0,1)$, $\rho\in(0,1)$, and $\tilde{J}_D$ is an estimator of the form given in \eqref{eq:ApproximatedThreshold}.
\end{problem}

Enforcing that the false alarm rate is close to the acceptable false alarm rate with a high probability, which means that  $\epsilon$ and $\rho$ are close to zero, will intuitively lead to threshold estimates that are not just trivially large to guarantee the acceptable false alarm rate.

\section{Finite sample guarantees}
\label{sec:FiniteSampleGuarantees}
In this section, we use three different approaches to obtain distribution-free finite sample guarantees that solve Problem~\ref{prob:FiniteGuarantees} and compare the finite guarantees with each other.
\subsection{Finite guarantees from the DKW inequality}
The first sample guarantee is based on the DKW inequality \cite{DKWInequality}, which gives us probabilistic bounds on how close the empirical distribution function $F_{N}(y_D)=\frac{1}{N}\sum_{i=1}^{N}\mathbf{1}_{\lbrace y_{D,i}\leq y_D\rbrace}$ is to the true cumulative distribution function $F(y_D)$,
\begin{align}
	\label{eq:DKWInequality}
	\mathrm{prob}\lbrace|F(y_D)-F_{N}(y_D)|\leq \epsilon \rbrace \geq 1-2e^{-2N\epsilon^2}.
\end{align}
While the DKW inequality gives us bounds for the whole CDF, we will evaluate it only at the point of interest, which is $y_D=\tilde{J}_D$, to obtain our finite guarantees. 
Note that the result was presented in our previous paper \cite{UmsonstACC21}, but is restated for the sake of completeness.
\begin{proposition}[Proposition 3 in \cite{UmsonstACC21}]
	\label{prop:FiniteGuaranteesDKW}
	Assume we have $N$ i.i.d. samples, $\lbrace y_{D,i}\rbrace_{i=1}^{N}$, of the detector output $y_D$ and $\gamma\in\mathbb{Q}$ such that $\gamma=\frac{n_1}{n_2}$ is an irreducible fraction.
	A solution to Problem~\ref{prob:FiniteGuarantees}  is given by $N=\left\lceil \frac{\ln(2\rho^{-1})}{2\epsilon^2n_2}\right\rceil n_2$ if $\beta\in[0,1)$ in \eqref{eq:ApproximatedThreshold}.
\end{proposition}
\begin{proof}
	First, we choose $N=kn_2$, where $k\in\mathbb{N}$, such that $N\gamma=kn_1$ is an integer as well.
	Further, with $\beta\in[0,1)$ we obtain that $F_{N}(\tilde{J}_D)=\gamma$. Evaluating the DKW inequality \eqref{eq:DKWInequality} at $y_D=\tilde{J}_D$ leads then to 
	\begin{align*}
		\mathrm{prob}\lbrace|F(\tilde{J}_D)-\gamma|\leq \epsilon \rbrace \geq 1-2e^{-2kn_2\epsilon^2}.
	\end{align*}
	Finally, we set $\rho=2e^{-2kn_2\epsilon^2}$ and solving for $k$ leads to $k=\left\lceil \frac{\ln(2\rho^{-1})}{2\epsilon^2n_2}\right\rceil$, which concludes the proof.
\end{proof}

Since the DKW inequality takes the whole CDF into account, this finite guarantee can be conservative.
Here, we are only interested in the point $y_D=\tilde{J}_D$ and not the complete probability distribution.
Therefore, we propose two more finite guarantees that evaluate the CDF at $y_D=\tilde{J}_D$ in the following.

\subsection{Finite guarantees from Vysochanskij-Petunin inequality}
Instead of focusing on the whole probability distribution as in the previous section, we now use the statistics of order statistics to determine a finite sample guarantee which utilizes the Vysochanskij-Petunin inequality \cite{VysPetInequality}.
The Vysochanskij-Petunin inequality is given by
\begin{align}
	\label{eq:VysInequality}
	\mathrm{prob}\lbrace |X-\mathbb{E}\lbrace X\rbrace|\geq \epsilon\rbrace \leq \frac{4\mathrm{Var}\lbrace X\rbrace}{9\epsilon^2}
\end{align}
if $3\epsilon^2\geq 8\mathrm{Var}\lbrace X\rbrace$, where $X$ is a unimodal random variable with a finite mean and variance.
Since for the Vysochanskij-Petunin inequality the expected value and variance of a random variable are needed, we introduce the expected value and variance of the CDF of a random variable evaluated at the $m$th order statistic.

\begin{lemma}
	\label{lem:MeanAndVarianceOfOrderStatistic}

Let $y_{D,(m)}$ be the $m$th order statistics of $N$ i.i.d. samples with CDF $F(\cdot)$. Then $F(y_{D,(m)})$ has a unimodal beta distribution with parameters $m$ and $N+1-m$, i.e., $F(y_{D,(m)})\sim\mathrm{Beta}(m,N+1-m)$, and the expected value and the variance of $F(y_{D,(m)})$ are given by
	\begin{align}
		\label{eq:MeanOrderStatistics}
		\mathbb{E}\lbrace F(y_{D,(m)})\rbrace = \frac{m}{N+1} 
	\end{align}
	and
	\begin{align}
		\label{eq:VarianceOrderStatistics}
		\mathrm{Var}\lbrace  F(y_{D,(m)}) \rbrace=\frac{m(N+1-m)}{(N+1)^2(N+2)},
	\end{align}
	respectively.

\end{lemma}
\begin{proof}
	From Chapter~2 in \cite{OrderStatistics}, we know that $F(y_{D,(m)})\sim\mathrm{Beta}(m,N+1-m)$. Hence, \eqref{eq:MeanOrderStatistics} and \eqref{eq:VarianceOrderStatistics} are the expected value and variance of the beta distribution with parameters $m$ and $N+1-m$, respectively. 
	Further, since $m\in\lbrace1,\ldots,N\rbrace$, both parameters of the beta distribution are larger than or equal to one, which indicates that the beta distribution is unimodal (see Chapter~2 in \cite{HandbookOfBetaDistribution}).
\end{proof}
Interestingly, neither the expected value nor the variance of the CDF at the $m$th order statistic depend on the distribution itself.
This is used in the following to determine a distribution-free finite sample guarantee.
\begin{theorem}
\label{thm:FiniteGuaranteesVysv}
Assume we have $N$ i.i.d. samples $\lbrace y_{D,i}\rbrace_{i=1}^{N}$ of the detector output $y_D$ and let $\gamma\in\mathbb{Q}$ such that $\gamma=\frac{n_1}{n_2}$ is its irreducible fraction. 
A solution to Problem~\ref{prob:FiniteGuarantees} is given by $N=kn_2-1$, where
\begin{align}
	\label{eq:ConditionOnkInVys}
	k=\left\lceil\frac{1}{n_2}\left(\frac{4\gamma(1-\gamma)}{9\rho\epsilon^2}-1\right)\right\rceil,
\end{align} 
if $4\gamma(1-\gamma)>9\rho\epsilon^2$, $6\rho\leq 1$, and $\tilde{J}_{D}=y_{D,(\lfloor N\gamma\rfloor+1)}$.
\end{theorem}
\begin{proof}
Let $\tilde{J}_D=y_{D,(m)}$. Then the false alarm rate of this threshold is given by $1-F(y_{D,(m)})$, which is a random variable that depends on the samples obtained. 
Hence, we use the Vysochanskij-Petunin  inequality \eqref{eq:VysInequality} and Lemma~\ref{lem:MeanAndVarianceOfOrderStatistic} to obtain
\begin{align*}
\mathrm{prob}\left\lbrace \left|F(y_{D,(m)})-\frac{m}{N+1}\right|\geq \epsilon\right\rbrace \leq \frac{4m(N+1-m)}{9\epsilon^2(N+1)^2(N+2)}.
\end{align*}
Next, with $k\in\mathbb{N}$ we set $N=kn_2-1$ and $m=\lfloor N\gamma\rfloor+1=kn_1$ such that $\frac{m}{N+1}=\gamma$, which leads to
\begin{align*}
\mathrm{prob}\left\lbrace \left|F(y_{D,(\lfloor N \gamma\rfloor+1)})-\gamma\right|\geq \epsilon\right\rbrace \leq \frac{4\gamma(1-\gamma)}{9\epsilon^2(kn_2+1)}.
\end{align*}
By introducing $\rho=\frac{4\gamma(1-\gamma)}{9\epsilon^2(kn_2+1)}$ and solving for $k$ we obtain \eqref{eq:ConditionOnkInVys} by making sure that $k$ is an integer. Next, to guarantee that $k\geq 1$, we need to introduce the condition $4\gamma(1-\gamma)>9\rho\epsilon^2$.
Finally, for the Vysochanskij-Petunin inequality \eqref{eq:VysInequality} to hold we need $3\epsilon^2\geq 8\mathrm{Var}\lbrace y_{D,(\lfloor N\gamma\rfloor+1)}\rbrace$ for the determined sample size $N$. 
Let us look at an upper bound for the variance first,
\begin{align*}
	\mathrm{Var}\lbrace y_{D,(\lfloor N\gamma\rfloor+1)}\rbrace&=\frac{\gamma(1-\gamma)}{\left\lceil\frac{1}{n_2}\left(\frac{4\gamma(1-\gamma)}{9\rho\epsilon^2}-1\right)\right\rceil n_2+1}\\
	&\leq \frac{\gamma(1-\gamma)}{\frac{1}{n_2}\left(\frac{4\gamma(1-\gamma)}{9\rho\epsilon^2}-1\right) n_2+1}=\frac{9}{4}\epsilon^2\rho.
\end{align*}
From this upper bound we obtain that if $6\rho\leq 1$ then $3\epsilon^2\geq 8\mathrm{Var}\lbrace y_{D,(\lfloor N\gamma\rfloor+1)}\rbrace$ holds.
\end{proof}
We would like to point out that the condition $\rho\leq \frac{1}{6}$ is not too restrictive, since $\rho$ is typically chosen to be small to achieve the guarantees with high probability.
Instead of using the Vysochanskij-Petunin inequality one could also use Chebyshev's inequality \cite{ChebyshevReference}, which leads to more relaxed conditions on $\rho$ and $\epsilon$ at the expense of having more conservative sample guarantees.

\subsection{Finite guarantees from confidence intervals of the beta distribution}
In the previous section, we determined finite guarantees based on the Vysochanskij-Petunin inequality, which uses the expected value and variance of $F(y_{D,(m)})$.
In this section, we will use that $F(y_{D,(m)})$ has a beta distribution and determine sample guarantees based on confidence intervals of the beta distribution.
We begin by introducing a result on the confidence interval of a beta distributed random variable.
\begin{lemma}
\label{lem:ExactConfIntBeta}
Let $X$ be distributed according to a beta distribution with parameters $m$ and $N+1-m$, i.e., $X\sim\mathrm{Beta}(m,N+1-m)$, and let $\hat{\gamma}=\frac{m}{N}\in(0,1)$. Then
	\begin{align*}
		\mathrm{prob}\lbrace \hat{\gamma}-\epsilon_l\leq X \leq \hat{\gamma}+\epsilon_u\rbrace\geq 1-\rho,
	\end{align*}
	where $\rho\in(0,1)$,
	\begin{align*}
	\epsilon_u&=\frac{\sqrt{\hat{\gamma}-\hat{\gamma}^2}}{\sqrt{N}}z_{\frac{\rho}{2}}+\frac{2(0.5-\hat{\gamma})z_{\frac{\rho}{2}}^2-1-\hat{\gamma}}{3N},\\
	\epsilon_l&=\frac{\sqrt{\hat{\gamma}-\hat{\gamma}^2}}{\sqrt{N}}z_{\frac{\rho}{2}}-\frac{2(0.5-\hat{\gamma})z_{\frac{\rho}{2}}^2-1-\hat{\gamma}}{3N},
	\end{align*}
	 up to order $\mathcal{O}(N^{-\frac{3}{2}})$, and $z_{\frac{\rho}{2}}$ is the upper $\frac{\rho}{2}$-quantile of the standard Gaussian distribution.

\end{lemma}
\begin{proof}
	From Lemma~1 and Theorem~1 in \cite{ExactConfidenceIntBinomial}, we obtain that $\mathrm{prob}\lbrace X\leq \gamma_q \rbrace=\rho_q$ holds for
	\begin{align*}
		\gamma_q=\hat{\gamma}-\frac{\sqrt{\hat{\gamma}-\hat{\gamma}^2}}{\sqrt{N}}z_{\rho_q}+\frac{2(0.5-\hat{\gamma})z_{\rho_q}^2-1-\hat{\gamma}}{3N},
	\end{align*} 
	up to order $\mathcal{O}(N^{-\frac{3}{2}})$.
	Using $\rho_q=\frac{\rho}{2}$ and $\rho_q=1-\frac{\rho}{2}$ to obtain $\gamma_L$ and $\gamma_U$, respectively, we determine that
	\begin{align*}
		\mathrm{prob}\lbrace \gamma_L \leq X \leq \gamma_U \rbrace=1-\rho.
	\end{align*}
	Finally, since $z_{1-\frac{\rho}{2}}=-z_{\frac{\rho}{2}}$, we obtain $\epsilon_l$ and $\epsilon_u$ by arithmetic operations, which concludes the proof.
\end{proof}
Note that the bounds $\epsilon_l$ and $\epsilon_u$ in Lemma~\ref{lem:ExactConfIntBeta} are functions of $\hat{\gamma}$, $\rho$, and $N$. 
For the sake of readability, we omit the parameters of these bounds.
Further, note that Lemma~\ref{lem:ExactConfIntBeta} provides a potentially asymmetric confidence interval, i.e., $\epsilon_l\neq\epsilon_u$, depending on the value of $\hat{\gamma}$ and $\rho$.
\begin{theorem}
	\label{thm:SampleGuaranteesExact}
	Assume $\gamma\in[0.5,1)$ is a rational number, such that $\gamma=\frac{n_1}{n_2}$ is the irreducible fraction of $\gamma$, and that we have $N$ i.i.d. samples of $y_D$, $\lbrace y_{D,i}\rbrace_{i=1}^{N}$.
	Under these assumptions, a solution to Problem~\ref{prob:FiniteGuarantees} is given by $N=kn_2$, and
	\begin{multline}
		\label{eq:ExactKForConfInterval}
		k=\left\lceil\left(\frac{z_{\frac{\rho}{2}}\sqrt{\gamma-\gamma^2}}{2\epsilon\sqrt{n_2}}\right.\right.\\
	\left.\left.+\sqrt{\left(\frac{z_{\frac{\rho}{2}}\sqrt{\gamma-\gamma^2}}{2\epsilon\sqrt{n_2}}\right)^2+\frac{2(\gamma-0.5)z_{\frac{\rho}{2}}^2}{3n_2\epsilon}+\frac{1+\gamma}{3n_2\epsilon}}\right)^2\right\rceil,
	\end{multline}
	if $\tilde{J}_{D}=y_{D,(N\gamma)}$, where  $z_{\frac{\rho}{2}}$ is the upper $\frac{\rho}{2}$-quantile of the standard Gaussian distribution.
\end{theorem}
\begin{proof}
First, recall that $F(y_{D,(m)})\sim\mathrm{Beta}(m,N+1-m)$ from Lemma~\ref{lem:MeanAndVarianceOfOrderStatistic} and note that by setting $m=N\gamma$ we have $\hat{\gamma}=\gamma$, since $N=kn_2$.
Next, by assuming that $\gamma\in[0.5,1)$, we determine that $\epsilon_l\geq\epsilon_u$.
This shows us that $[\gamma-\epsilon_l,\gamma+\epsilon_u]\subseteq[\gamma-\epsilon_l,\gamma+\epsilon_l]$ for $\gamma\in[0.5,1)$.
Therefore, with $\tilde{J}_D=y_{D,(N\gamma)}$, the true probability of triggering no alarm is given by $F(\tilde{J}_D)$ such that we obtain the confidence interval ${\mathrm{prob}(|F(\tilde{J}_D)-\gamma|\leq \epsilon_l)\geq 1- \rho}$ from Lemma~\ref{lem:ExactConfIntBeta}.
By setting $N=kn_2$ and $\epsilon_l=\epsilon$ in Lemma~\ref{lem:ExactConfIntBeta} we solve for $\sqrt{k}$ and obtain
\begin{equation}
	\begin{aligned}
		\sqrt{k}&=\frac{z_{\frac{\rho}{2}}\sqrt{\gamma-\gamma^2}}{2\epsilon\sqrt{n_2}}\\
	&\quad\pm\sqrt{\left(\frac{z_{\frac{\rho}{2}}\sqrt{\gamma-\gamma^2}}{2\epsilon\sqrt{n_2}}\right)^2+\frac{2(\gamma-0.5)z_{\frac{\rho}{2}}^2}{3n_2\epsilon}+\frac{1+\gamma}{3n_2\epsilon}}.
	\end{aligned}
\end{equation}
Since $\sqrt{k}\geq 0$, we discard the solution with a negative sign. After squaring and rounding up to the next larger integer, we obtain \eqref{eq:ExactKForConfInterval}. This concludes the proof.
\end{proof}
Note that it is not restrictive to only consider $\gamma\in[0.5,1)$ in Theorem~\ref{thm:SampleGuaranteesExact}.
This is because
\begin{align*}
\mathrm{prob}\lbrace|F(\tilde{J}_D)-\gamma|\leq \epsilon\rbrace&=\mathrm{prob}\lbrace|1-F(\tilde{J}_D)-(1-\gamma)|\leq \epsilon\rbrace\\
&=\mathrm{prob}\lbrace|F(\hat{J}_D)-(1-\gamma)|\leq \epsilon\rbrace,
\end{align*}
where $\hat{J}_D$ is the threshold that approximates a false alarm rate of $\gamma$, which exists due to the continuity of the CDF.
Thus, the integer $k$ obtained from \eqref{eq:ExactKForConfInterval} for a certain $\gamma\in[0.5,1)$ is the same as for $1-\gamma$.
\subsection{Discussion and comparison of the bounds}
\label{sec:GuaranteeComparison}
In this section we will discuss the three finite sample guarantees obtained previously, investigate their scaling in the parameters $\epsilon$ and $\rho$, and compare the sample sizes.  
For a given $\gamma$, $\epsilon$, and $\rho$, let $N_{\mathrm{DKW}}$, $N_{\mathrm{VP}}$, and $N_{\mathrm{beta}}$ be the sample sizes obtained from Proposition~\ref{prop:FiniteGuaranteesDKW}, Theorem~\ref{thm:FiniteGuaranteesVysv}, and Theorem~\ref{thm:SampleGuaranteesExact}, respectively.
First, we want to discuss how to choose $\epsilon$. 
A reasonable bound for $\epsilon$ is $\epsilon\leq\min(\gamma,1-\gamma)$, because this choice results in $[\gamma-\epsilon,\gamma+\epsilon]\subseteq[0,1]$. 
This means that the $\pm\epsilon$-band around the acceptable false alarm rate contains only reasonable false alarm rates, i.e., false alarm rates in $[0,1]$. 
Further, with that choice we have $\frac{4}{9}\gamma(1-\gamma)\geq \frac{4}{9}\min(\gamma,1-\gamma)^2\geq\frac{4}{9}\epsilon^2>\frac{1}{6}\epsilon^2\geq\rho\epsilon^2$ if $\rho\leq \frac{1}{6}$.
Hence, if $\rho\leq \frac{1}{6}$ the condition on both $\rho$ and $\epsilon$ in Theorem~\ref{thm:FiniteGuaranteesVysv} is fulfilled for this bound of $\epsilon$.

Next, we want to investigate the scaling of each guarantee in the parameters $\epsilon$ and $\rho$. 
All finite guarantees scale with $\epsilon^{-2}$, while they differ in their scaling in $\rho$. The guarantee scales with $\ln(\rho^{-1})$, $\rho^{-1}$ and $z_{\frac{\rho}{2}}$ when it is obtained from Proposition~\ref{prop:FiniteGuaranteesDKW}, Theorem~\ref{thm:FiniteGuaranteesVysv}, and Theorem~\ref{thm:SampleGuaranteesExact}, respectively. 
Typically, we desire high probability guarantees such that $\rho$ is close to zero.
It follows that Theorem~\ref{thm:SampleGuaranteesExact} has the best scaling for $\rho\in(0,0.5]$ since then $z_{\frac{\rho}{2}}<\ln(\rho^{-1})<\rho^{-1}$.

Now, we compare the sample sizes from the three different finite guarantees when $\rho=0.05$ and $\epsilon=0.01$ for $\gamma\in[0.01,0.99]$.
Fig.~\ref{fig:ComparisonOfSampleSize} shows the results of the comparison. 
\begin{figure}
\centering
\includegraphics[width=0.5\textwidth]{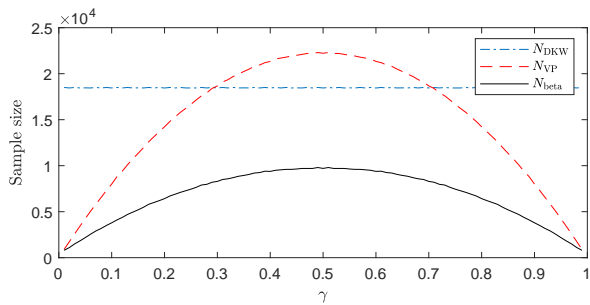}
\caption{We compare the sample sizes obtained from the DKW inequality (dash-dotted line), from the Vysochanskij-Petunin inequality (dashed line) and from the beta confidence intervals (solid line), where $\gamma\in[0.01,0.99]$, $\rho=0.05$, and $\epsilon=0.01$.}
\label{fig:ComparisonOfSampleSize}
\end{figure}
We see that the required sample size obtained from the DKW inequality is almost constant for different $\gamma$.
This is expected because the DKW inequality applies for the whole cumulative distribution function.
Further, we observe that Theorem~\ref{thm:SampleGuaranteesExact} produces the smallest sample sizes out of the three finite guarantees for all investigated $\gamma$.
We also observe that $N_{\mathrm{VP}}\leq N_{\mathrm{DKW}}$ if $\gamma\not\in(0.294,0.706)$. 
So only for $\gamma$ close to $0$ and $1$ the finite guarantees from Theorem~\ref{thm:FiniteGuaranteesVysv} perform better than the finite guarantees from Proposition~\ref{prop:FiniteGuaranteesDKW}.

Both Fig.~\ref{fig:ComparisonOfSampleSize} and the scaling in $\rho$ and $\epsilon$ show us that the sample sizes we obtain from Theorem~\ref{thm:SampleGuaranteesExact} are smaller than the sample sizes from Proposition~\ref{prop:FiniteGuaranteesDKW} and Theorem~\ref{thm:FiniteGuaranteesVysv} for all investigated $\gamma$.
Furthermore, according to \cite{ExactConfidenceIntBinomial}, the confidence bounds $\epsilon_u$ and $\epsilon_l$ are (nearly) exact if $N\geq 40$, since the influence of the higher order terms disappears. 
\section{Numerical examples}
\label{sec:NumExamples}
In the first part of section, we evaluate the sampling guarantees numerically for three different detector output distributions.
In the second part, we use detector output data from an experimental setup to tune the threshold of a cumulative sum (CUSUM) detector without resetting. The code to reproduce these results can be found at \url{https://github.com/DavidUmsonst/FiniteSampleGuaranteesForQuantileEstimation}.

\subsection{Evaluation of sample guarantees}
We begin by evaluating the finite sample guarantees. 
The idea is to approximate the threshold of the anomaly detector with \eqref{eq:ApproximatedThreshold} based on $N$ samples, $\lbrace y_{D,i}\rbrace_{i=1}^{N}$, of the detector output and we do this approximation $N_T=1000$ times to obtain $N_T$ different approximations $\tilde{J}_D$ of $J_D$.
Then we draw $10^6$ new samples of the detector output and calculate the empirical false alarm rate of each threshold.

Here, we investigate three different cases.
First, we assume that the detector output has a $\chi^2(4)$ distribution, where we have four degrees of freedom. The samples are then taken i.i.d. from a $\chi^2(4)$ distribution.
Second, we assume that the detector output has a L\'evy distribution and obtain i.i.d. samples for the output from a L\'evy distribution.
Third, we assume the samples are taken from the trajectory of a non-parametric CUSUM detector \cite{RuthsMultivariate} without resetting, given by
\begin{align}
	\label{eq:NonParamCUSUM}
	y_D(k+1)=\max(0,y_D(k)+\|r(k)\|_2^2-\delta),
\end{align}
where $y_D(0)=0$, $\delta=6$, and $r(k)$ is the input of the detector and is drawn i.i.d. from a four-dimensional, zero-mean, multivariate Gaussian distribution at each time step $k$.
Furthermore, we use $y_{D,i}=y_D(i)$ to obtain the samples, that is, the samples are the trajectory of the CUSUM detector.
Therefore, the i.i.d. assumption on the samples is not fulfilled in this case.

For our simulation, we choose $\gamma=0.95$, $\epsilon=0.01$, and $\rho=0.05$, which means that the empirical false alarm rate should be in the interval $[0.04, 0.06]$ with a probability of $95\,\%$.
Moreover, Section~\ref{sec:GuaranteeComparison} showed us that for these values of $\gamma$, $\epsilon$, and $\rho$, $N_{\mathrm{beta}}<N_{\mathrm{VP}}<N_{\mathrm{DKW}}$.
Therefore, we investigate only the smallest and largest sample sizes, $N_{\mathrm{beta}}$ and $N_{\mathrm{DKW}}$, respectively, for the sake of clarity.
In addition to that, we set $\beta=0$ when approximating the threshold with \eqref{eq:ApproximatedThreshold}.

\begin{figure}
\centering
\includegraphics[width=0.5\textwidth]{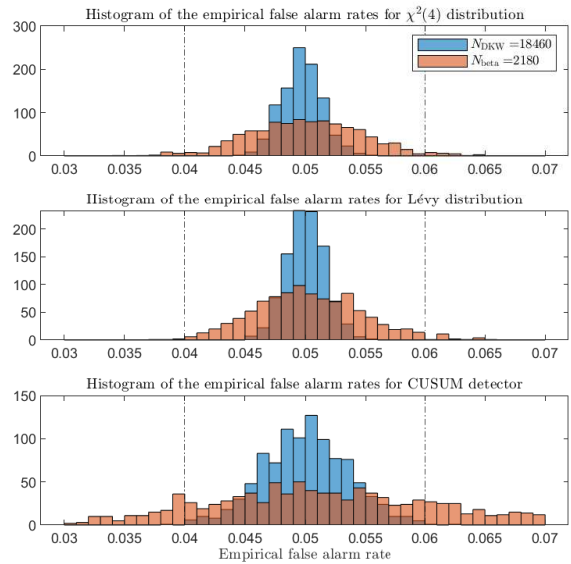}
\caption{The empirical false alarm rate from $N_T=1000$ thresholds is evaluated over a data set of $10^6$ samples for a $\chi^2(4)$ distribution (upper plot), a L\'evy distribution (center plot), and samples obtained from a CUSUM detector (lower plot).}
\label{fig:EmpFARForThreeDistributionsWithTwoSampleSizes}
\end{figure}

Fig.~\ref{fig:EmpFARForThreeDistributionsWithTwoSampleSizes} shows the histogram of the empirical false alarm rate for the $\chi^2(4)$ distribution (upper plot), the L\'evy distribution (center plot), and samples obtain from the CUSUM detector (lower plot) for each of the two sample sizes investigated.
The two upper plots confirm our theoretical results, where the empirical false alarm rate lies inside the desired interval with a probability larger than $95\,\%$. More specifically, when $N=N_{\mathrm{beta}}$ and $y_D$ has a $\chi^2$ distribution only $4.2\,\%$, that is, $42$ out of $1000$, of the empirical false alarm rates are outside of the desired interval.
Similarly, when $N=N_{\mathrm{beta}}$ and $y_D$ has a L\'evy distribution only $2.4\,\%$ of the empirical false alarm rates are outside of the desired interval.
If $N=N_{\mathrm{DKW}}$ none of the empirical false alarm rates are outside of the desired interval.
Therefore, we see that the sample size provided by Theorem~\ref{thm:SampleGuaranteesExact} are very close to the desired guarantees of at most $5\,\%$ of false alarm rates being outside the desired interval, while the sample size provided by Proposition~\ref{prop:FiniteGuaranteesDKW} has much better probabilistic guarantees but has also a conservative amount of samples.
Furthermore, the two upper plots in Fig.~\ref{fig:EmpFARForThreeDistributionsWithTwoSampleSizes} verify that the guarantees hold for both light-tailed and heavy-tailed distributions.

However, in the lower plot, where the samples are not i.i.d., we observe that $34.1\,\%$ of the empirical false alarm rates lie outside of the desired interval when $N=N_{\mathrm{beta}}$ is used.
For $N=N_{\mathrm{DKW}}$ only $0.3\,\%$ of the empirical false alarm rates lie outside the desired interval. 
Hence, we see that the probabilistic guarantees are not fulfilled, when $N=N_{\mathrm{beta}}$.
Although in this case the choice of $N=N_{\mathrm{DKW}}$ provides enough conservatism to outweigh the effect of the non-i.i.d. sampling, in principle there is no guarantee in either case.
Therefore, the sample size obtained in Theorem~\ref{thm:SampleGuaranteesExact} is the smallest sample size that fulfills the probabilistic guarantees given $N_{\mathrm{beta}}$ i.i.d. samples, but it is sensitive to violations of the i.i.d. assumption.

One reasonable approach to deal with correlated data is to collect more data than the minimum amount for i.i.d. data and randomly sample an approximately i.i.d. subset of the correlated data.
Proposition~\ref{prop:FiniteGuaranteesDKW}, Theorem~\ref{thm:FiniteGuaranteesVysv}, and Theorem~\ref{thm:SampleGuaranteesExact} can help with determining the size of the approximately i.i.d. subset of correlated data, but not the minimum size of the correlated data needed.
Determining the minimum size of the correlated data set for threshold estimation needs further investigation and is an avenue for future work.
This approach of dealing with correlated data is used in the next section, when we obtain samples from our experimental setup.

\subsection{Tuning with real data}
In this last section, we evaluate the sample guarantees provided by Theorem~\ref{thm:SampleGuaranteesExact} with data obtained from an experimental setup, where a real process is controlled.
The process used is the Temperature Control Lab (TCLab), which consists of two heaters and one temperature sensor for each heater (for more details see \cite{TCLabModelBenchmark}).
In the experimental setup we control both heaters to have a temperature of $40\,{}^\circ \mathrm{C}$. 
We implement an LQG controller, which is based on a linearized and discretized model around the steady-state temperature $40\,{}^\circ \mathrm{C}$ for both heaters.
Here we use a sampling time of $1\,\mathrm{s}$ to obtain measurements.
From the TCLab's sensors we obtain two measurements, $y(k)\in\mathbb{R}^2$, and with the predicted sensor measurements $\hat{y}(k)$ from the steady-state Kalman filter inside the LQG controller the residual signal is $\bar{r}(k)=y(k)-\hat{y}(k)$.
Here, $r(k)=\hat{\Sigma}^{-\frac{1}{2}}(\bar{r}(k)-\hat{\mu})$ serves as the input to \eqref{eq:NonParamCUSUM},
where $\hat{\mu}$ and $\hat{\Sigma}$ are approximations of the mean and the covariance matrix of the residual signal, respectively, used to normalize $\bar{r}(k)$, $\delta=3$ in \eqref{eq:NonParamCUSUM}, and the initial state of the CUSUM detector is again set to zero.
Note that, to the best of our knowledge, there exists no closed-form solution for the threshold of a CUSUM detector that guarantees an acceptable false alarm rate $1-\gamma$.
Therefore, we will use $N_{\mathrm{beta}}$ samples of the detector to approximate the threshold according to \eqref{eq:ApproximatedThreshold} with $\beta=0$ for an acceptable false alarm rate of $1-\gamma$.

We let the experiment run for $15379\,\mathrm{s}$ to gather data.
Since the detector output is assumed to be a random variable with a fixed distribution (see Assumption~\ref{assum:PDFofDetectorOutput}), we only work with samples from the steady state of the system, which is approximately reached for $k\geq 780$.
We obtain $\hat{\mu}$ and $\hat{\Sigma}$ from the first 1000 samples of the residual signal in steady state, i.e., the data set $\mathcal{D}_{\mu\Sigma}=\lbrace r(k)\rbrace_{k=780}^{1779}$.
With $\hat{\mu}$ and $\hat{\Sigma}$ available we determine the detector output $y_D(k+1)$ such that we use the data set $\mathcal{D}=\lbrace y_D(k)\rbrace_{k=1779}^{15379}$ of detector outputs to estimate the threshold, where $y_D(1779)=0$, and validate the empirical false alarm rate of the threshold estimate.

In the following, we apply the approach outlined at the end of the previous section.
This means that, for a given $\gamma$, we randomly choose $N_{\mathrm{beta}}$ samples from $\mathcal{D}$ without replacement to obtain an approximately i.i.d. data set for determining $\tilde{J}_D$. 
Since we only have a finite amount of samples in $\mathcal{D}$, the remaining samples of $\mathcal{D}$ are compared to the threshold to produce the empirical false alarm rate induced by $\tilde{J}_D$. 
Repeatedly applying this approach is known as Repeated Training/Test Splits \cite{AppliedPredictiveModelling}, which lets us simultaneously estimate and evaluate $\tilde{J}_D$ with a finite data set $\mathcal{D}$.
This is done repeatedly to obtain $N_T=10000$ empirical false alarm rates for a given $\gamma$.
For tuning the detector, we choose $\epsilon=0.01$ and $\rho=0.05$ and investigate nine different values for $\gamma$, that is, $\lbrace\gamma_i\rbrace_{i=1}^9$, where $\gamma_i=0.95+(i-1)\cdot 0.005$.
For these values of $\epsilon$ and $\rho$, the amount of samples for testing is at least $11421$ for each of the investigated values of $\gamma$.

Fig.~\ref{fig:BoxplotAccFARoverEmpFAR} shows a box plot of the empirical false alarm rates over the investigated acceptable false alarm rates $1-\gamma$.
We observe that the median value of the empirical false alarm rate is almost exactly located at the acceptable false alarm rate for all investigated values of $\gamma$.
Furthermore, the box plots are concentrated around the acceptable false alarm rate.
Only $5\,\%$ of the empirical false alarm rates should lie outside of the $\pm\epsilon$-band (shaded area in Fig.~\ref{fig:BoxplotAccFARoverEmpFAR}) around the acceptable false alarm rate, since $\rho=0.05$.
Here, the largest percentage of empirical false alarm rates outside the $\pm\epsilon$-band is $4.77\,\%$ for an acceptable false alarm rate of $0.045$, i.e., $\gamma=0.955$.
Hence, although the data in $\mathcal{D}$ is strongly correlated due the CUSUM detector dynamics \eqref{eq:NonParamCUSUM}, the theoretical guarantees from Theorem~\ref{thm:SampleGuaranteesExact} hold. 
The reason for that is that due to the random selection of $N_\mathrm{beta}$ samples from $\mathcal{D}$ to determine $\tilde{J}_D$, it is unlikely that several adjacent samples, e.g., $y_D(2000)$, $y_D(2001)$, and $y_D(2002)$, are chosen such that the samples in the training set are not highly correlated anymore.
\begin{figure}
\centering
\includegraphics[width=0.5\textwidth]{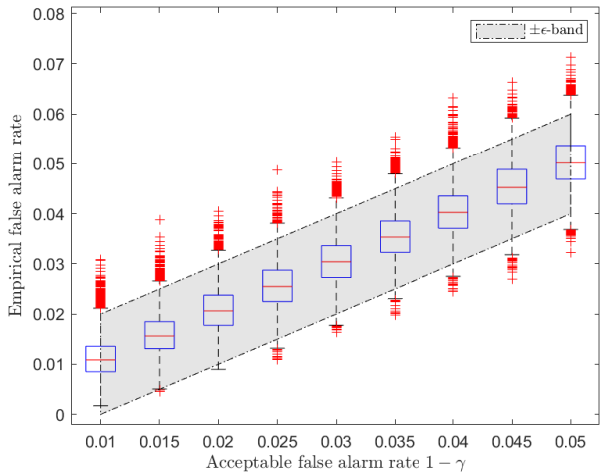}
\caption{A box plot of the empirical false alarm rate over the acceptable false alarm rate $1-\gamma$ when $N_T=10000$ threshold estimates with corresponding their empirical false alarm rates are obtained from the data set $\mathcal{D}$ via Repeated Training/Test Splits.}
\label{fig:BoxplotAccFARoverEmpFAR}
\end{figure}
Note that the samples in the test set are still strongly correlated. 
This shows that the guarantees for the threshold determined from the uncorrelated samples holds for the correlated samples in the test set as well.

Finally, we look at the histogram of the empirical false alarm rates for $\gamma=0.95$ (see Fig.~\ref{fig:EmpFARHistogramCUSUMRealData}).
We observe that in contrast to the lower plot in Fig.~\ref{fig:EmpFARForThreeDistributionsWithTwoSampleSizes} the histogram for $N=N_{\mathrm{beta}}$ is now concentrated around $1-\gamma=0.05$ and only $4.64\,\%$ of the empirical false alarm rates are located outside of the desired interval marked by the vertical dash-dotted lines.
This demonstrates that by selecting training data using random sampling (as opposed to sequential sampling) we are able to employ the sample guarantees in Theorem~\ref{thm:SampleGuaranteesExact} to find a threshold for the highly correlated output of the CUSUM detector.

\begin{figure}
\centering
\includegraphics[width=0.5\textwidth]{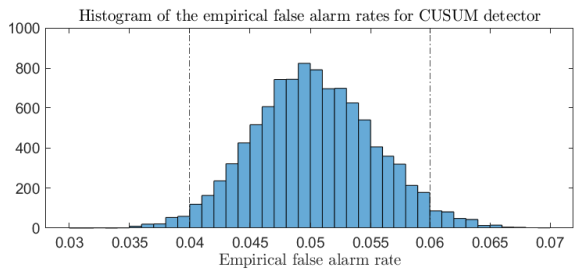}
\caption{The empirical false alarm rate from $N_T=10000$ thresholds is evaluated over a data set $\mathcal{D}$ when Repeated Training/Test Splits are used to determine $\tilde{J}_D$, where the area between the vertical dash-dotted lines represents the desired $\pm\epsilon$-band.}
\label{fig:EmpFARHistogramCUSUMRealData}
\end{figure}

\section{Conclusion}
\label{sec:Conclusion}

In this work, we considered the tuning of detector thresholds and pointed out the equivalence of the detector threshold and a specific quantile of the detector output distribution. 
We derived three different finite guarantees for the estimation of a quantile. 
The first is based on the DKW inequality, which takes the whole cumulative density function into account.
The second is based on the Vysochanskij-Petunin inequality and uses the expected value and variance of the CDF evaluated at a specific order statistic to determine the sample guarantees.
The third is based the confidence interval of a beta distribution and utilizes a closed-form solution of the confidence interval bounds.

When comparing the three guarantees, we saw that the third guarantee has the best scaling in the confidence parameter and leads to the smallest sample size.
Simulations showed that the i.i.d. assumption is important and can lead to violations of the guarantees. 
However, we showed in our experimental setup that using random instead of sequential samples to tune a threshold can be an effective way to avoid the adverse effects of a non-i.i.d detector output without changing the sequential implementation of the detector operation.

Avenues for future work involve relaxing the assumptions we have made in this work. We aim to develop extended results to directly take the non-i.i.d. nature of the detector output into account rather than to use the indirect random sampling approach we propose here. 
We would also like to be able to address non-stationary detector output distributions, including results that provide guidance on real-time threshold selection.

\bibliographystyle{IEEEtran}
\bibliography{IEEEexample}

\end{document}